\documentclass[aip,jmp,reprint,onecolumn,groupedaddress]{revtex4-1}

\usepackage{amssymb,amsmath,amsthm}

\usepackage[bookmarksnumbered,colorlinks,linkcolor=blue,citecolor=blue,urlcolor=blue]{hyperref}
\usepackage{hypcap}
\usepackage{amsthm}
\usepackage{enumerate}
\usepackage{colonequals}
\usepackage{overpic}

\numberwithin{equation}{section}

\DeclareMathOperator{\adj}{adj}
\DeclareMathOperator{\spn}{span}
\DeclareMathOperator{\Res}{Res}

\renewcommand{\Im}{\mathop{\mathrm{Im}}}

\newcommand{\C}{\mathbb{C}}
\newcommand{\R}{\mathbb{R}}

\renewcommand{\sc}{{\rm sc}}
\newcommand{\Conf}{\mathcal{C}}

\newcommand{\defeq}{\colonequals}

\newtheorem{theorem}{Theorem}
\newtheorem{lemma}[theorem]{Lemma}
\newtheorem*{theorem*}{Theorem}

\newcommand{\eq}[1]{\hyperref[eq:#1]{(\ref*{eq:#1})}}
\renewcommand{\sec}[1]{\hyperref[sec:#1]{Section~\ref*{sec:#1}}}
\newcommand{\app}[1]{\hyperref[app:#1]{Appendix~\ref*{app:#1}}}
\newcommand{\fig}[1]{\hyperref[fig:#1]{Figure~\ref*{fig:#1}}}
\newcommand{\thm}[1]{\hyperref[thm:#1]{Theorem~\ref*{thm:#1}}}
\newcommand{\lem}[1]{\hyperref[lem:#1]{Lemma~\ref*{lem:#1}}}

\begin{document}

\title{Levinson's theorem for graphs II}

\author{Andrew M. Childs}

\author{David Gosset}

\address{
Department of Combinatorics \& Optimization and
Institute for Quantum Computing,
University of Waterloo,
200 University Avenue West, Waterloo, Ontario,
Canada}

\begin{abstract}
We prove Levinson's theorem for scattering on an $(m+n)$-vertex graph with $n$ semi-infinite paths each attached to a different vertex, generalizing a previous result for the case $n=1$.  This theorem counts the number of bound states in terms of the winding of the determinant of the S-matrix.  We also provide a proof that the bound states and incoming scattering states of the Hamiltonian together form a complete basis for the Hilbert space, generalizing another result for the case $n=1$.
\end{abstract}
\maketitle

\section{Introduction}

Continuous-time quantum walks on graphs were introduced by Farhi and Gutmann as a framework for developing new quantum algorithms.\cite{1998PhRvA..58..915F}  This idea was subsequently applied to give an example of exponential speedup by quantum walk \cite{2002quant.ph..9131C} and an optimal quantum algorithm for evaluating game trees.\cite{2007quant.ph..2144F}  Recently, it was shown that even a highly restricted model of continuous-time quantum walk is universal for quantum computation.\cite{2009PhRvL.102r0501C}

A key feature of the quantum walks considered in Refs.~\onlinecite{2009PhRvL.102r0501C,2007quant.ph..2144F,1998PhRvA..58..915F} is that the dynamics can be understood using scattering theory. For certain infinite graphs, one can construct an analog of standard quantum scattering theory, defining an S-matrix as well as eigenstates of the Hamiltonian corresponding to scattering of a wave packet at some momentum.\cite{2009PhRvL.102r0501C,1998PhRvA..58..915F}

Reference \onlinecite{2011JMP....52h2102C} proved a version of Levinson's theorem for continuous-time quantum walks on graphs obtained by attaching a semi-infinite path to a single vertex of a finite graph. As in Levinson's original work,\cite{1949PhRv...75.1445L} the theorem proven in Ref.~\onlinecite{2011JMP....52h2102C} gives a relation between the phase of the reflection coefficient and the number of bound states supported by the Hamiltonian. In this paper we prove a more general version of Levinson's theorem for graphs.

We begin by reviewing scattering theory and continuous-time quantum walks on graphs. At the end of this section we summarize our new results and their relationship to previous work.

\subsection*{Quantum walk and scattering theory}

Quantum walk is a quantum mechanical analog of classical random walk. In this paper we consider continuous-time quantum walk,\cite{1998PhRvA..58..915F} in which time evolution of the quantum walker occurs via the Schr{\"o}dinger equation with a time-independent Hamiltonian. 

A simple example is the quantum walk on an infinite path. The Hilbert space is spanned by basis vectors 
\[
\{|x\rangle\colon x\in \mathbb{Z}\}
\]
and the Hamiltonian is
\begin{equation}
H=\sum_{x=-\infty}^{\infty}\left(|x\rangle\langle x+1|+|x+1\rangle\langle x|\right).\label{eq:line_walk}
\end{equation}
Starting in a particular vertex state $|y\rangle$ and evolving with the Hamiltonian $H$, the quantum state after time $t$ is 
\[
e^{-iHt}|y\rangle.
\]

To understand the dynamics, consider the analogy with the Schr{\"o}dinger equation for a free particle in one dimension. There the Hamiltonian is 
\[
H_{\text{free}}=\frac{p^{2}}{2m}
\]
where the momentum operator $p$ can be written in the position basis as $-i\hbar\frac{d}{dx}$. The Hamiltonian \eq{line_walk} is simply the finite difference approximation to $H_{\text{free}}$, up to an overall constant and a term proportional to the identity. As in the case of the free particle, the Hamiltonian is diagonalized in the basis of momentum eigenstates
\[
|\tilde k\rangle = \sum_{x=-\infty}^{\infty}e^{-ikx}|x\rangle,\quad k\in[-\pi,\pi].
\]
 These states satisfy 
\[
\langle \tilde k|\tilde p\rangle=2\pi\delta(k-p)
\]
where $\delta$ is the Dirac delta function. Such a state $|\tilde k\rangle$ is an eigenstate of the Hamiltonian with energy $E=2\cos(k)$. The interpretation of these states is analogous to the interpretation of the momentum states for the free particle in one dimension. We can imagine preparing a wave packet, that is to say, a superposition of momentum states with momenta close to some value $k$. For $k\in(-\pi,0)$ the wave packet moves to the left under the Hamiltonian evolution and for $k\in(0,\pi)$ it moves to the right. Indeed, the group velocity of such a wave packet is
\[
-\frac{dE}{dk}=2\sin(k).
\]
Note that, unlike the case of a free particle, the dispersion relation $E(k)=2\cos(k)$ has the property that its phase velocity $-{E}/{k}$ does not always have the same sign as its group velocity. However, it is the group velocity that determines the overall motion of a wave packet that is narrowly peaked in momentum space.

More interesting solutions to the one-dimensional Schr{\"o}dinger equation arise when a potential term $V(x)$ is added to the free Hamiltonian $H_{\text{free}}$. In the familiar case where $V(x)$ goes to zero sufficiently fast as $x\rightarrow\pm\infty$, the eigenstates of the Hamiltonian
\[
\frac{p^{2}}{2m}+V(x)
\]
are of two types: bound states and scattering states. The bound states are normalizable states with amplitudes that asymptotically go to zero as $x\rightarrow\pm\infty$, whereas the scattering states asymptotically approach momentum states in this limit.

One can define a quantum walk on any graph $G$. The Hilbert space has basis vectors labeled by the vertices of $G$. If $G$ is an undirected graph, we take the Hamiltonian $H_{G}$ to be the adjacency matrix of the graph. More generally, we can consider quantum walks on directed graphs where each edge $i\rightarrow j$ is weighted by a complex number $w_{ij}$, subject to the constraint that $w_{ji} = w_{ij}^*$ for all $i,j$. Then define the Hamiltonian to be the Hermitian operator
\begin{equation}
H_{G} \defeq \sum_{i,j} w_{ij} |j\rangle\langle i|.\label{eq:graph_ham}
\end{equation}

\begin{center}
\begin{figure}
\capstart
\begin{overpic}{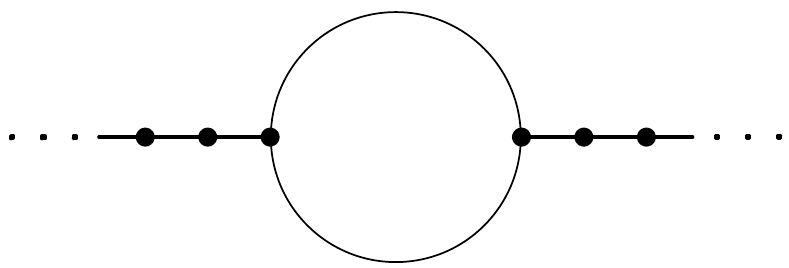}
\put(46.5,13.5){\huge$\widehat{G}$}
\end{overpic}
\caption{Quantum walk on an infinite path with an obstacle represented by the directed
graph $\widehat{G}$. The graph $\widehat{G}$ includes the two vertices that
are attached to the semi-infinite paths and may also include some
internal vertices (not pictured). }
\label{fig:line}
\end{figure}
\end{center}

To consider scattering on graphs, we construct a graph $G$ by attaching semi-infinite paths to some of the vertices of a finite graph $\widehat{G}$. The associated Hamiltonian $H_G$ (defined through equation \eq{graph_ham}) is equal to the adjacency matrix of the graph on the semi-infinite paths, and is equal to $H_{\hat{G}}$ within the original graph $\widehat{G}$. The case of two semi-infinite paths, depicted in \fig{line}, is closely analogous to the case of scattering off a (finite-range) one-dimensional potential. As in the one-dimensional Schr{\"o}dinger equation, eigenstates can be either scattering states or bound states. The scattering states are eigenstates of this Hamiltonian, but they can also be viewed as describing the dynamics of a wave packet that is prepared on one semi-infinite path and then allowed to evolve according to the Hamiltonian. Each scattering state has components that can be interpreted as an incident wave, a reflected wave, and a wave that is transmitted through the obstacle $\widehat{G}$.

In this paper we discuss quantum walks on graphs obtained by attaching $n$ semi-infinite paths to a graph $\widehat{G}$, as depicted in \fig{graph}. Now we can prepare an incoming wave packet on any of the $n$ semi-infinite paths and allow it to scatter off of the obstacle $\widehat{G}$. Associated with each incoming momentum $k\in(-\pi,0)$ and each semi-infinite path $j\in\{1,2,\ldots,n\}$, there is a scattering eigenstate $|\sc_{j}(k)\rangle$. The states $\{|\sc_{j}(k)\rangle\colon j\in\{1,2,\ldots,n\}\}$ can be compactly described by an $n\times n$ unitary matrix $S$ called the S-matrix.

\subsection*{Outline of the paper}

In \sec{eigenstates} we describe the bound and scattering eigenstates of the Hamiltonian, as well as the S-matrix, following Refs.~\onlinecite{2009PhRvL.102r0501C,2009PhRvA..80e2330V}. \thm{completeness} of \sec{eigenstates} (proved in the Appendix) shows that the incoming scattering states along with the bound states form a complete basis for the Hilbert space. As in previous work,\cite{2011JMP....52h2102C,goldstone} our proofs rely on analytic continuation of the S-matrix. We define this analytic continuation in \sec{Analytic-Continuation-of}. In \sec{Levinson's-Theorem} we prove Levinson's theorem for graphs of the form shown in \fig{graph}. Our proof relies on a technical lemma that is proven in \sec{Proof-of-Lemma}.

\subsection*{Relation to previous work}

Levinson's theorem relates the number of bound states of the Hamiltonian to the winding number of $\det(S(k))$ as $k$ is varied from $-\pi$ to $\pi$ (i.e., the number of times the phase of the determinant wraps around the interval $[0,2\pi)$ as $k$ varies). In the case $n=1$, our theorem is simpler and slightly stronger than the theorem proven in Ref.~\onlinecite{2011JMP....52h2102C}, removing a minor technical requirement on the graph $\hat{G}$. Our results can also be viewed as generalizing previous work which discusses discrete versions of Levinson's theorem on the half-line with a boundary condition at one end\cite{case:594,hinton:754} (rather than a general weighted finite graph $\widehat{G}$ as in Ref.~\onlinecite{2011JMP....52h2102C}).

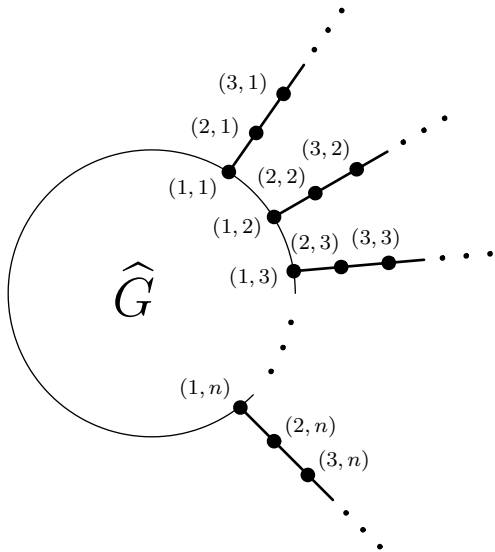
\begin{figure}
\capstart
\begin{centering}
\begin{overpic}{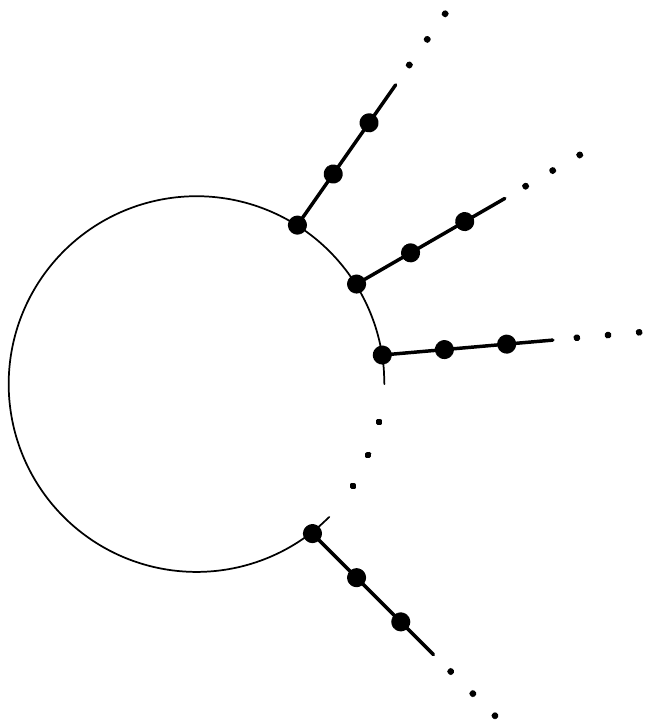}
\put(20,43){\huge$\widehat{G}$}
\put(30,65){\footnotesize$(1,1)$}
\put(34,76){\footnotesize$(2,1)$}
\put(39,84){\footnotesize$(3,1)$}
\put(38,58){\footnotesize$(1,2)$}
\put(46,67){\footnotesize$(2,2)$}
\put(54,72){\footnotesize$(3,2)$}
\put(41,49){\footnotesize$(1,3)$}
\put(52,55){\footnotesize$(2,3)$}
\put(63,56){\footnotesize$(3,3)$}
\put(32,29){\footnotesize$(1,n)$}
\put(51,22){\footnotesize$(2,n)$}
\put(57,16){\footnotesize$(3,n)$}
\end{overpic}
\end{centering}
\caption{Graphs $G$ that we consider in this paper consist of a finite, weighted graph $\widehat{G}$ with $m+n$ vertices attached to $n$ semi-infinite paths. The $m$ internal vertices of $\widehat{G}$ are not pictured. We label the vertices on the semi-infinite paths $(x,j)$ where $x\in\{1,2,3,\ldots\}$ and $j\in\{1,2,\ldots,n\}.$}
\label{fig:graph}
\end{figure}

\section{Eigenstates of the Hamiltonian\label{sec:eigenstates}}

Consider the Hamiltonian describing a quantum walk on a graph $G$ as shown in \fig{graph}. Let $\widehat{G}$ have $m+n$ vertices, with $n$ vertices attached to semi-infinite paths and $m$ ``internal'' vertices. The Hamiltonian $H_{\widehat{G}}$ associated with the graph $\widehat{G}$ is defined as in equation \eq{graph_ham}. For convenience, we denote by $\widehat{H}$ the $(m+n)\times(m+n)$ matrix of $H_{\widehat{G}}$ in the basis of the $m+n$ vertices of $\widehat{G}$. (In general, we use a hat to denote restriction to the finite graph). We can write 
\[
\widehat{H}=\begin{pmatrix}
A & B^{\dagger}\\
B & D
\end{pmatrix}
\]
where the matrix $A$ is $n\times n$, $B$ is $m\times n$, and $D$ is $m\times m$. The full Hamiltonian $H$ associated with the graph $G$ is $H_{\widehat{G}}$ plus a term connecting adjacent vertices on the semi-infinite paths: 
\[
H=H_{\widehat{G}} + \sum_{j=1}^{n}\sum_{x=1}^{\infty}\left(|x,j\rangle\langle x+1,j|+|x+1,j\rangle\langle x,j|\right).
\]

\subsection*{Bound states}

Bound states are the normalizable eigenstates of the Hamiltonian, and therefore have amplitudes on each semi-infinite path that go to zero as the distance along the line increases. Each bound state $|\phi\rangle$ has the form 
\begin{equation}
\langle x,j|\phi\rangle=\alpha_{j}z^{x-1}\label{eq:bound_states}
\end{equation}
 on the semi-infinite paths (for $j\in\{1,\ldots,n\}$ and $x\in\{1,2,3,\ldots\}$) for some $\alpha_{j}\in\C$ and $z\in(-1,1)\setminus\{0\}$ (the boundary cases where $z=\pm1$  will be discussed later). Any state of this form automatically satisfies 
\[
\langle x,j|H|\phi\rangle=\left(z+\frac{1}{z}\right)\langle x,j|\phi\rangle
\]
for $x\geq2$ and $j\in\{1,\ldots,n\}$. A bound state must also satisfy
the eigenvalue equation at the $m+n$ vertices of the graph. Let $\vec{\alpha}$
be a column vector with entries $\alpha_{j}$ for $j=1,\ldots,n$ and
let $\vec{\beta}$ be a column vector of the $m$ amplitudes of the
state $|\phi\rangle$ that are inside the graph and not on the semi-infinite
paths. Then the eigenvalue equation can be written as
\[
\begin{pmatrix}
A & B^{\dagger}\\
B & D
\end{pmatrix}\begin{pmatrix}
\vec{\alpha}\\
\vec{\beta}
\end{pmatrix}+z\begin{pmatrix}
\vec{\alpha}\\
0
\end{pmatrix}=\left(z+\frac{1}{z}\right)\begin{pmatrix}
\vec{\alpha}\\
\vec{\beta}
\end{pmatrix}.
\]
We write this in a more compact form using the operator 
\begin{equation}
\gamma(z) \defeq
\begin{pmatrix}
zA-1 & zB^{\dagger}\\
zB & zD-z^{2}-1
\end{pmatrix}.
\label{eq:gamma}
\end{equation}
Then the eigenvalue equation is
\[
\gamma(z)\begin{pmatrix}
\vec{\alpha}\\
\vec{\beta}
\end{pmatrix}=0.
\]
Given any $(m+n)$-dimensional vector $|\widehat{v}\rangle=\begin{pmatrix}
\vec{\alpha}\\
\vec{\beta}
\end{pmatrix}$ satisfying
\[
\gamma(z)|\widehat{v}\rangle=0
\]
for some $z\in(-1,1)$, there is an associated normalized bound state $|\phi_{v}\rangle$
defined through
\begin{align*}
\langle x,j|\phi_{v}\rangle & =N_{v} \alpha_{j}z^{x-1} \text{ for }x\in\{1,2,\ldots\}\text{ and }j\in\{1,2,\ldots,n\}\\
\langle w|\phi_{v}\rangle & =N_{v} \beta_{w} \text{ for internal vertices of the graph }w\in\{1,\ldots,m\}.
\end{align*}
 The normalizing constant $N_{v}$ is 
\begin{align}
N_{v} & = \left(\frac{\| \vec{\alpha}\| ^{2}}{1-z^{2}}+\| \vec{\beta}\| ^{2}\right)^{-\frac{1}{2}}.\label{eq:normalizing}
\end{align}

\subsection*{Confined and unconfined bound states}

Some bound states may have zero amplitude on each of the $n$ semi-infinite paths, corresponding to $\vec{\alpha}=\vec{0}$ in the previous discussion. These have been called confined bound states\cite{2011JMP....52h2102C} or bound states of the second kind\cite{2009PhRvA..80e2330V}. A confined bound state $|\psi_{c}\rangle$ has nonzero amplitude only on the $m$ internal vertices of the graph that are not on the semi-infinite paths. Writing 
\[
|\widehat{\psi}_{c}\rangle=\begin{pmatrix}
0\\
\vec{\beta_{c}}
\end{pmatrix}
\]
for the restriction of $|\psi_{c}\rangle$ to the vertices of the graph, we see that this vector must satisfy the eigenvalue equation
\[
\begin{pmatrix}
A & B^{\dagger}\\
B & D
\end{pmatrix}|\widehat{\psi}_{c}\rangle=\lambda_{c}|\widehat{\psi}_{c}\rangle
\]
for some $\lambda_{c}\in\mathbb{R}$. This implies $D\vec{\beta}_{c}=\lambda_{c}\vec{\beta}_{c}$ and $B^{\dagger}\vec{\beta}_{c}=0$. Define the projector onto the semi-infinite paths
\[
P_{n} \defeq \sum_{x=1}^{\infty}\sum_{q=1}^{n}|x,q\rangle\langle x,q|
\]
and its restriction to the graph 
\[
\widehat{P_{n}} \defeq \begin{pmatrix}
1 & 0\\
0 & 0
\end{pmatrix}
\]
which projects onto the $n$ vertices connected to the semi-infinite paths. A confined bound state $|\psi_{c}\rangle$ satisfies $P_{n}|\psi_{c}\rangle=0$ and 
\begin{align}
\gamma(z)|\widehat{\psi}_{c}\rangle & = \left(z\lambda_{c}-z^{2}-1\right)|\widehat{\psi}_{c}\rangle.
\label{eq:gammacbs}
\end{align}
The roots of the polynomial $z\lambda_{c}-z^{2}-1$ are $z_{c}$ and
$z_{c}^{-1}$ where 
\[
z_{c}=\frac{1}{2}\left(\lambda_{c}-\sqrt{\lambda_{c}^{2}-4}\right).
\]
If $|\lambda_{c}|>2$ then both roots are real, with one inside the unit circle and one outside. The vector space spanned by the confined bound states with energies $|\lambda_{c}|>2$ is therefore 
\begin{equation}
\Conf_{>} \defeq \spn\left\{ |\psi\rangle\colon \text{$\langle\psi|P_{n}|\psi\rangle=0$ and  $\gamma(z)|\widehat{\psi}\rangle=0$ for some $z$ with $|z|<1$}\right\} .\label{eq:S1}
\end{equation}
If $|\lambda_{c}|<2$ then both roots are on the unit circle and are complex conjugates of one another. The vector space spanned by the confined bound states with energies $|\lambda_{c}|<2$ is 
\begin{equation}
\Conf_{<} \defeq \spn\left\{ |\psi\rangle\colon\text{$\langle\psi|P_{n}|\psi\rangle=0$ and $\gamma(z)|\widehat{\psi}\rangle=0$ for some $z$ with $|z|=1$ and $z\notin\{-1,1\}$}\right\} .\label{eq:S2}
\end{equation}
If $\lambda_{c} = \pm2$ then there is one repeated root equal
to $\pm1$. The corresponding vector space spanned by the confined bound states
is
\begin{equation}
\Conf_{=} \defeq \spn\left\{ \text{$|\psi\rangle\colon \langle\psi|P_{n}|\psi\rangle=0$ and $\gamma(z)|\widehat{\psi}\rangle=0$  for $z \in \{-1,1\}$}\right\} .\label{eq:S3}
\end{equation}

 In general, the Hamiltonian can have confined bound states as well as bound states with nonzero amplitudes on the semi-infinite paths. We can always form an orthonormal basis of the bound states of the Hamiltonian consisting of confined bound states in $\Conf \defeq \Conf_{<}\oplus \Conf_{=}\oplus \Conf_{>}$ and bound states in its orthogonal complement $\Conf^\perp$. We refer to bound states in $\Conf^\perp$ as ``unconfined bound states.''

\subsection*{Half-bound states}

Half-bound states\cite{hinton:754} are unnormalizable states that are ``almost'' bound states. They are eigenstates of the Hamiltonian taking the form \eqref{eq:bound_states} on the semi-infinite lines with $z\in\{-1,1\}$ (where $\vec{\alpha}\neq\vec{0})$. 

\subsection*{Scattering states and the S-matrix}

For each $k\in(-\pi,0)$ we define a set of $n$ incoming scattering states $\{|\sc_{j}(k)\rangle\colon j\in\{1,\ldots,n\}\}$ which have the form
\begin{align*}
\langle x,j'|\sc_{j}(k)\rangle & = e^{-ikx}\delta_{j'j}+e^{ikx}S_{j'j}(e^{ik}) \\
 &= \frac{1}{z^{x}}\delta_{j'j}+z^{x}S_{j'j}(z)
\end{align*}
on the semi-infinite lines, where $z=e^{ik}$. Such a state has energy $z+\frac{1}{z}=2\cos k.$ (Here $z$ is on the unit circle; later we discuss the analytic continuation of the S-matrix to other values of $z$.) The label $j$ indicates the semi-infinite line on which the state is incoming and $x=1,2,3,\ldots$ indexes the distance along this line (with $|1,j\rangle$ corresponding to the vertex where the $j$th line connects to the $(m+n)$-vertex graph).

We now write the eigenvalue equations that determine the amplitudes of $|\sc_{j}(k)\rangle$ at internal vertices of the graph. For each $j$, write these $m$ amplitudes as a column vector $\vec{\psi}_{j}(z)$, and collect these column vectors into a matrix 
\[
\Psi(z) \defeq \begin{pmatrix}
\vec{\psi}_{1}(z) & \vec{\psi}_{2}(z) & \cdots & \vec{\psi}_{n}(z)\end{pmatrix}.
\]
The following matrix equation determines $\Psi(z)$ as well as the scattering matrix $S(z)$: 
\begin{align*}
\begin{pmatrix}
A & B^{\dagger}\\
B & D
\end{pmatrix}\begin{pmatrix}
\left(z^{-1}+zS(z)\right)\\
\Psi(z)
\end{pmatrix}+\begin{pmatrix}
\left(z^{-2}+z^{2}S(z)\right)\\
0
\end{pmatrix}
=\left(z + \frac{1}{z}\right)\begin{pmatrix}
\left(z^{-1}+zS(z)\right)\\
\Psi(z)
\end{pmatrix}.
\end{align*}
The lower part of this equation says
\[
\Psi(z)=\frac{1}{\frac{1}{z}+z-D}\left(z^{-1}B+zBS(z)\right),
\]
which determines each $\vec{\psi}_{j}(z)$ in terms of the scattering matrix $S(z)$. The upper part determines the scattering matrix. We find
\[
A\left(z^{-1}+zS(z)\right)+B^{\dagger}\Psi(z)+\left(z^{-2}+z^{2}S(z)\right)=\left(z+\frac{1}{z}\right)\left(z^{-1}+zS(z)\right),
\]
which gives 
\begin{align}
S(z) & =-Q(z)^{-1}Q(z^{-1})\label{eq:S_matrix}
\end{align}
where
\begin{align*}
Q(z) &\defeq 1-z\left(A+B^{\dagger}\frac{1}{\frac{1}{z}+z-D}B\right).
\end{align*}

Recall that (for now) we are restricting our attention to values of $z$ on the unit circle. For such values, the S-matrix is unitary. To see this, note that $Q(z^{-1})=Q(z)^{\dagger}$ and $[Q(z),Q(z^{-1})]=0$, so
\begin{align}
S(z)^{\dagger} & =-Q(z^{-1})^{\dagger}\left(Q(z)^{-1}\right)^{\dagger} \nonumber\\
 & =-Q(z)Q(z^{-1})^{-1} \nonumber\\
 & =-Q(z^{-1})^{-1}Q(z) \nonumber\\
 & =S(z^{-1}). \label{eq:sdaggerinv}
\end{align}
We can see from equation \eq{S_matrix} that $S(z^{-1})=S(z)^{-1}$. This establishes unitarity of $S(z)$ on the unit circle.


\subsection*{A complete basis of eigenstates}

The bound states and incoming scattering states form a complete basis for the entire infinite-dimensional Hilbert space. Let $\{|\psi_{c}\rangle\colon c=1,\ldots,n_{c}\}$ be an orthonormal basis of the confined bound state subspace $\Conf$ such that for each $c$, $P_{n}|\psi_{c}\rangle=0$ and 
\[
\begin{pmatrix}
A & B^{\dagger}\\
B & D
\end{pmatrix}|\widehat{\psi}_{c}\rangle=\lambda_{c}|\widehat{\psi}_{c}\rangle.
\]
Furthermore, let $\{|\phi_{b}\rangle\colon b=1,\ldots,n_{b}\} \subset \Conf^\perp$ be a basis for the unconfined bound state subspace. In \app{completeness} we prove the following theorem, generalizing the $n=1$ case:\cite{goldstone}
\begin{theorem}
\label{thm:completeness}Let $v$ and $w$ be any two
vertices of the graph $G$. Then 
\[
\langle v|\left(\int_{-\pi}^{0}\frac{dk}{2\pi}|\sc_{j}(k)\rangle\langle\sc_{j}(k)|+\sum_{b=1}^{n_{b}}|\phi_{b}\rangle\langle\phi_{b}|+\sum_{c=1}^{n_{c}}|\psi_{c}\rangle\langle\psi_{c}|\right)|w\rangle=\delta_{vw}.
\]
\end{theorem}

\section{Analytic continuation of the S-matrix\label{sec:Analytic-Continuation-of}}

We would like to analytically continue the scattering matrix $S(z)$ from the unit circle to the rest of the complex plane. To do this, we rewrite equation \eq{S_matrix} so that each matrix element of $S(e^{ik})$ is manifestly a rational function of $z=e^{ik}$.  Recall the definition of $\gamma(z)$ in equation \eq{gamma}. As can be verified by direct calculation, we have
\[
\gamma(z)=\begin{pmatrix}
1 & zB^{\dagger}\\
0 & zD-z^{2}-1
\end{pmatrix}\begin{pmatrix}
-Q(z) & 0\\
\frac{z}{zD-z^{2}-1}B & 1
\end{pmatrix}
\]
and 
\[
\gamma(z)^{-1}=\begin{pmatrix}
-Q(z)^{-1} & 0\\
\frac{z}{zD-z^{2}-1}BQ(z)^{-1} & 1
\end{pmatrix}\begin{pmatrix}
1 & -B^{\dagger}\left(D-z-\frac{1}{z}\right)^{-1}\\
0 & \frac{1}{z}\left(D-z-\frac{1}{z}\right)^{-1}
\end{pmatrix}.
\]
From these expressions we obtain 
\begin{equation}
-\gamma(z)^{-1}\gamma(\tfrac{1}{z})=\begin{pmatrix}
S(z) & 0\\
\frac{1}{z}\Psi(z) & -\frac{1}{z^{2}}
\end{pmatrix}.\label{eq:S_analytic_continuation}
\end{equation}

Now consider this equation for $z\in\C$ (no longer restricting to the unit circle). Each matrix element of $\gamma(z)$ is a polynomial in $z$. It is invertible everywhere in $\C$ except at a set of points determined by the roots of the polynomial $\det(\gamma(z)$). We can write its inverse as 
\[
\gamma(z)^{-1}=\frac{1}{\det(\gamma(z))}\text{adj}(\gamma(z))
\]
where $\adj(\gamma(z))$ is the adjugate matrix of $\gamma(z)$. The matrix elements of $\adj(\gamma(z))$ are polynomials in $z$. Hence the entries of $\gamma(z)^{-1}$ are rational functions of $z$, and so are the entries of $-\gamma(z)^{-1}\gamma(\tfrac{1}{z})$. We therefore define the analytic continuation $S(z)$ through equation \eq{S_analytic_continuation}, so $S(z)$ is the upper left $n\times n$ submatrix of $-\gamma(z)^{-1}\gamma(\tfrac{1}{z})$. 
\begin{lemma}
\label{lem:S_is_meromorphic}
Let $S(z)$ be defined through equation \eq{S_analytic_continuation}. Then each matrix element of $S(z)$ is a rational function of $z$. Furthermore, if $z_{0}$ is a pole of some matrix element of $S(z)$, then either $z_{0}=0$ or $\det(\gamma(z_{0}))=0$.
\end{lemma}

\begin{proof}
It follows from the above discussion that each matrix element of $S(z)$ is a rational function of $z$. By equation \eq{gamma}, the entries of the matrix $z^{2}\gamma(\tfrac{1}{z})$ are polynomials in $z$. Hence each matrix element of $S(z)$ is a rational function of $z$ with denominator $z^{2}\det(\gamma(z))$.
\end{proof}

\section{Levinson's theorem\label{sec:Levinson's-Theorem}}

Levinson's theorem counts the number of bound states. Let us now make this more precise. Using expresions \eq{S1}, \eq{S2}, and \eq{S3}, we define the number of confined bound states as
\[
n_{c} \defeq \dim \Conf = \dim \Conf_{>}+\dim \Conf_{<}+\dim \Conf_{=}.
\]
Furthermore, define the number of unconfined bound states as
\begin{align*}
n_{b} & \defeq \sum_{\substack{x\in(-1,1) \\ \det\gamma(x)=0}} \dim\left(\spn\left\{ |\psi\rangle \colon |\psi\rangle\in \Conf_{>}^{\perp} \text{ and }\gamma(x)|\widehat{\psi}\rangle=0\right\} \right)
\end{align*}
and the number of (unconfined) half-bound states as
\begin{align*}
n_{h} & \defeq \sum_{x=\pm1}\dim\left(\spn\left\{ |\psi\rangle \colon |\psi\rangle\in \Conf_{=}^{\perp} \text{ and }\gamma(x)|\widehat{\psi}\rangle=0\right\} \right).
\end{align*}
For our purposes, the half-bound states are only counted as half a bound state each. In other words, we consider the ``number of bound states'' to be 
\[
n_{c}+n_{b}+\frac{1}{2}n_{h}.
\]

We now give another formula for the number of bound states. Define
\[
W(z) \defeq \det(\gamma(z)).
\]
Note that $W(z)$ is a polynomial in $z$. Denote the multiset of roots of $W$ by
\[
\{z_{1},\ldots,z_{k}\}
\]
where $k$ is the degree of $W$ (each root appears in the above list
a number of times equal to its multiplicity). Let
\begin{align*}
\alpha_{1} & \defeq \left|\left\{ i\colon\left|z_{i}\right|<1\right\} \right|\\
2\alpha_{2} & \defeq \left|\left\{ i\colon\left|z_{i}\right|=1 \text{ and }  z_{i}\notin\{-1,1\}\right\} \right|.\\
\alpha_{3} & \defeq \left|\left\{ i\colon z_{i}\in\left\{ -1,1\right\} \right\} \right|.
\end{align*}
The following lemma relates the number of bound states to $\alpha_{1}$, $\alpha_{2}$, and $\alpha_{3}$.

\begin{lemma}
\label{lem:1}
With the definitions given above,
\begin{align*}
\alpha_{1} & =n_{b}+\dim \Conf_{>}\\
\alpha_{2} & =\dim \Conf_{<}\\
\alpha_{3} & =n_{h}+2\dim \Conf_{=},
\end{align*}
so $\alpha_{1}+\alpha_{2}+\frac{1}{2}\alpha_{3}=n_{b}+n_{c}+\frac{1}{2}n_{h}.$
\end{lemma}

The proof of \lem{1} is given in \sec{Proof-of-Lemma}. 

We recall some useful facts from complex analysis. Given a closed, positively-oriented curve $\kappa$ in the complex plane and a complex function $f(z)$ that is meromorphic in $\C$ and has no zeros or poles on $\kappa$, we define the winding number $w_{\kappa}(f)$ of $f$ around $\kappa$ to be the number of times the image of $\kappa$ wraps around the origin. In other words, it is the number of times the complex phase of $f$ wraps around the interval $[0,2\pi)$. The argument principle is a formula relating $w_{\kappa}(f)$ to the number of zeros and poles of $f$ inside the contour. It says that
\[
w_{\kappa}(f)=Z_{\kappa}(f)-P_{\kappa}(f)
\]
where $Z_{\kappa}(f)$ is the number of zeros of $f$ inside the contour $\kappa$ and $P_{\kappa}(f)$ is the number of poles of $f$ inside $\kappa$ (both counted with multiplicity). We also use the notation $Z_{\kappa\setminus\{a\}}(f)$ (respectively, $P_{\kappa\setminus\{a\}}(f)$) to indicate the number of zeros (respectively, poles) of $f$ inside $\kappa$ but excluding the point $a$.

We now prove Levinson's theorem:
\begin{theorem*}
The winding number of the determinant of the S-matrix around
the unit circle $\Gamma$ is
\[
w_{\Gamma}(\det(S))=2\left(m-n_{b}-n_{c}-\frac{1}{2}n_{h}\right).
\]
\end{theorem*}

This generalizes the main result of Ref.~\onlinecite{2011JMP....52h2102C}, which establishes the $n=1$ case.  Note that the approach of Ref.~\onlinecite{2011JMP....52h2102C} has a technical requirement on $\widehat{G}$, namely that $A\neq0$ or $B^{\dagger}B\neq1$ (in this case $A$ is $1\times1$ and $B$ is $m \times 1$). Our theorem has no such technical requirement, regardless of the value of $n$.

Also note that the form of Levinson's theorem depends on the conventions outlined in \sec{eigenstates} for the definition of the S-matrix. A different convention for the relative phases of the scattering states would modify the statement of the theorem.

\begin{proof}
Using equation \eq{S_analytic_continuation}, we find
\begin{align}
\det(S(z)) 
 &= (-1)^n z^{2m}\frac{W(\frac{1}{z})}{W(z)}\nonumber \\
 &= (-1)^n z^{2m}\left(\prod_{i=1}^{k}\frac{\frac{1}{z}-z_{i}}{z-z_{i}}\right)\nonumber \\
 &= (-1)^n z^{2m}\left(\prod_{i=1}^{k}\frac{z_{i}}{z}\right)\prod_{j=1}^{k}\left(\frac{\frac{1}{z_{j}}-z}{z-z_{j}}\right)\nonumber \\
 & =(-1)^n \left(\prod_{i=1}^{k}z_{i}\right)z^{2m-k}\frac{\prod_{j=1}^{k}\left(\frac{1}{z_{j}}-z\right)}{\prod_{j=1}^{k}\left(z-z_{j}\right)}.\label{eq:det_S}
\end{align}
(Note that $|z_{i}|>0$ for all $i$ because $\gamma(0)=-1$, so $W(0)=\left(-1\right)^{m+n}\neq0$.)

Although $W$ may have roots on the unit circle, $\det(S(z))$ does not have any zeros or poles on the unit circle. Indeed, $|{\det(S)}|=1$ on the unit circle since $S$ is unitary there. We can also see this explicitly from equation \eq{det_S}. Note that each root $z_{j}=\pm1$ is a root of both $W(z)$ and $W(\frac{1}{z})$ and the $\alpha_{3}$ factors corresponding to these roots cancel in the ratio
\begin{equation}
\frac{\prod_{j=1}^{k}\left(\frac{1}{z_{j}}-z\right)}{\prod_{j=1}^{k}\left(z-z_{j}\right)}\label{eq:ratio}
\end{equation}
appearing in the expression above. Similarly, for each root $z_{j}$ such that $|z_{j}|=1$ and $z_{j}\notin\{-1,1\}$, there is another root
$z_j^*=\frac{1}{z_{j}}$.
(For $z\in\R$, $\gamma(z)$ is Hermitian and $W(z)\in\R$, so any roots of $W$ with nonzero imaginary part must occur in complex conjugate pairs.)
These $2\alpha_{2}$ roots also cancel in the expression for $\det(S(z))$.

Similarly, if there are other roots $z_{j}$ with $0<|z_{j}|<1$ such that $\frac{1}{z_{j}}$ is also a root of $W(z)$ then the corresponding factors in the numerator and denominator of equation \eq{det_S} cancel each other. Let $q$ be the number of such roots $z_{j}$ (so that the total number of canceling factors in the ratio \eq{ratio} is $2q$). Then
\begin{align*}
P_{\Gamma\setminus\{0\}}(\det S) & =\alpha_{1}-q\\
Z_{\Gamma\setminus\{0\}}(\det S) & =\left(\text{degree of numerator}\right)-\left(\text{\# of zeros of numerator outside the unit circle}\right)\\
 & =\left(k-2q-2\alpha_{2}-\alpha_{3}\right)-\left(\alpha_{1}-q\right)\\
 & =k-q-\alpha_{1}-2\alpha_{2}-\alpha_{3},
\end{align*}
where ``numerator'' and ``denominator'' refer to the ratio in equation \eq{ratio} after common factors have been canceled.

The determinant of $S$ may also have a zero or a pole at $z=0$. From equation \eqref{eq:det_S} we see that the lowest-order term in the Laurent expansion about $z=0$ is $2m-k$, corresponding to a zero if this quantity is positive or a pole if it is negative.

Finally, the argument principle shows that
\begin{align*}
w_{\Gamma}\left(\det S\right) & =Z_{\Gamma\setminus\{0\}}\left(\det S\right)-P_{\Gamma\setminus\{0\}}\left(\det S\right)+2m-k,\\
 & =\left(k-q-\alpha_{1}-2\alpha_{2}-\alpha_{3}\right)-\left(\alpha_{1}-q\right)+2m-k\\
 &= 2m-2\alpha_1-2\alpha_2-\alpha_3 \\
 & =2\left(m-n_{b}-n_{c}-\frac{1}{2}n_{h}\right),
\end{align*}
where in the last line we have used \lem{1}.
\end{proof}

\section{Proof of Lemma 3 \label{sec:Proof-of-Lemma}} 

The operator 
\begin{equation}
\gamma(z)=\begin{pmatrix}
zA-1 & zB^{\dagger}\\
zB & zD-z^{2}-1
\end{pmatrix}
=z^2(\widehat{P_n}-1) + z\widehat{H} - 1
\label{eq:gammahp}
\end{equation}
has appeared in our discussions of the scattering and bound states of the Hamiltonian. In this section we first establish some technical properties of this operator that we use in the proof of \lem{1}. We use the following result of Kato (Theorem 6.1 and Section 6.2 of Ref.~\onlinecite{Kato:1966:PTL}).

\begin{theorem}
\label{thm:kato} Suppose $T_{0},T_{1},T_{2}$ are $N\times N$ Hermitian matrices and consider 
\[
T(x)=T_{0}+xT_{1}+x^{2}T_{2}
\]
as a function of the complex variable $x$. For each real $x$ there exists an orthonormal set of eigenvectors $\{|w_{i}(x)\rangle\colon i=1,2,\ldots,N\}$ of $T(x)$ which can be chosen to be holomorphic functions of $x$ on the real axis.
\end{theorem}

We now use this theorem to establish that the eigenvalues and eigenvectors of $\gamma(x)$ are smooth functions of $x$ for $x\in\R$. In fact, we show that one can choose a smooth basis for the eigenvectors that includes the confined bound states as basis vectors. In the proof of \lem{1} we use this fact to write $\det(S(x))$ (for $x\in\R$) as a product of two terms: one term that incorporates the contribution of the confined bound states and another term that comes from the unconfined bound states.

\begin{lemma}
\label{lem:diagonalize_gamma}Let $\{|\psi_{c}\rangle\}$ be an orthonormal basis of confined bound states as described in \thm{completeness}. For $x\in\R$, $\gamma(x)$ is Hermitian and there is an orthonormal basis $\{|\widehat{v_{i}}(x)\rangle\colon i=1,\ldots,m+n-n_{c}\}$ and eigenvalues $\{e_{i}(x)\colon i=1,\ldots,m+n-n_{c}\}$ that are holomorphic functions of $x$ on the real axis, such that 
\[
\gamma(x)
=\sum_{c=1}^{n_{c}}|\widehat{\psi}_{c}\rangle\langle\widehat{\psi}_{c}|\left(x\lambda_{c}-x^{2}-1\right)
+\sum_{i=1}^{m+n-n_{c}}e_{i}(x)|\widehat{v_{i}}(x)\rangle\langle\widehat{v_{i}}(x)| \quad\text{for $x\in\R$}.
\]
\end{lemma}

\begin{proof}
Write 
\[
\gamma(z)
=\sum_{c=1}^{n_{c}}|\widehat{\psi}_{c}\rangle\langle\widehat{\psi}_{c}|\left(z\lambda_{c}-z^{2}-1\right)
+\left(1-\sum_{c=1}^{n_{c}}|\widehat{\psi}_{c}\rangle\langle\widehat{\psi}_{c}|\right)\gamma(z)\left(1-\sum_{c=1}^{n_{c}}|\widehat{\psi}_{c}\rangle\langle\widehat{\psi}_{c}|\right).
\]
The second term can be written as 
\begin{align*}
\left(1-\sum_{c=1}^{n_{c}}|\widehat{\psi}_{c}\rangle\langle\widehat{\psi}_{c}|\right)\gamma(z)\left(1-\sum_{c=1}^{n_{c}}|\widehat{\psi}_{c}\rangle\langle\widehat{\psi}_{c}|\right) & = M_{0}+zM_{1}+z^{2}M_{2}
\end{align*}
for Hermitian matrices $M_{0},M_{1},M_{2}$ that are independent of $z$. The result follows by applying \thm{kato}. 
\end{proof}

We also use the following two lemmas in the proof of \lem{1}:
\begin{lemma}
\label{lem:deriv}If $e_{i}(x_{0})=0$ for some $x_{0}\in\R$ then 
\[
\frac{de_{i}}{dx}\bigg|_{x_{0}}=\left(\frac{1}{x_{0}}-x_{0}\right)+x_{0}\langle\widehat{v_{i}}(x_{0})|\widehat{P_{n}}|\widehat{v_{i}}(x_{0})\rangle.
\]
\end{lemma}

\begin{proof}
Note that $\gamma(x_{0})|\widehat{v_{i}}(x_{0})\rangle=0$ implies $x_{0}\neq0$ (since $\gamma(0)=-1)$. Dividing through by $x_{0}$ gives 
\[
\widehat{H}|\widehat{v_{i}}(x_{0})\rangle=\frac{1}{x_{0}}|\widehat{v_{i}}(x_{0})\rangle+x_{0}\left(1-\widehat{P_{n}}\right)|\widehat{v_{i}}(x_{0})\rangle.
\]
We have
\[
\frac{d\gamma}{dx}=\widehat{H} - 2x(1-\widehat{P_n}),
\]
so
\begin{align}
\frac{d\gamma}{dx}\bigg|_{x_{0}}|\widehat{v_{i}}(x_{0})\rangle & = \left[\frac{1}{x_{0}}-x_{0}\left(1-\widehat{P_{n}}\right)\right]|\widehat{v_{i}}(x_{0})\rangle.\label{eq:deriv}
\end{align}
Now 
\begin{align*}
\frac{d}{dx}e_{i}(x) & =\langle\widehat{v_{i}}(x)|\left(\frac{d}{dx}\gamma(x)\right)|\widehat{v_{i}}(x)\rangle+\left(\frac{d}{dx}\langle\widehat{v_{i}}(x)|\right)\gamma(x)|\widehat{v_{i}}(x)\rangle+\langle\widehat{v_{i}}(x)|\gamma(x)\left(\frac{d}{dx}|\widehat{v_{i}}(x)\rangle\right)\\
 & =\langle\widehat{v_{i}}(x)|\left(\frac{d}{dx}\gamma(x)\right)|\widehat{v_{i}}(x)\rangle
\end{align*}
 since $|\widehat{v_{i}}(x)\rangle$ is a normalized eigenvector of $\gamma(x)$ (this is sometimes called the Hellmann-Feynman theorem). Using equation
\eq{deriv}, 
\[
\frac{d}{dx}e_{i}(x)\bigg|_{x_{0}}=\left(\frac{1}{x_{0}}-x_{0}\right)+x_{0}\langle\widehat{v_{i}}(x_{0})|\widehat{P_{n}}|\widehat{v_{i}}(x_{0})\rangle. \qedhere
\]
\end{proof}

\begin{lemma}
~\label{lem:wzero}
\begin{enumerate}[(a)]
\item If $W(z)=0$ and $0\leq|z|<1$ then $\Im(z)=0$.
\item If $W(z)=0$, $|z|=1$, and $\Im(z)\neq 0$, then $\langle \widehat{\psi} |\widehat{P_{n}}|\widehat{\psi} \rangle=0$ for any $|\widehat{\psi}\rangle$ in the null space of $\gamma(z)$.
\end{enumerate}
\end{lemma}

\begin{proof}
First consider part (a). It is clearly true when $z=0$, so suppose $\left|z\right|>0$. The hypothesis $W(z)=0$ implies that there exists a normalized state $|\widehat{\psi}\rangle$ such that 
\begin{align*}
z^{-1}\gamma(z)|\widehat{\psi}\rangle & = \left[\widehat{H}-\left(z+\frac{1}{z}\right)+z\widehat{P_{n}}\right]|\widehat{\psi}\rangle
 = 0,
\end{align*}
so
\[
\langle\widehat{\psi}|\left[\widehat{H}-\left(z+\frac{1}{z}\right)+z\widehat{P_{n}}\right]|\widehat{\psi}\rangle=0.
\]
Writing $z=re^{i\phi}$ and taking the imaginary part of the above expression gives
\[
\sin\phi\left[r\langle\widehat{\psi}|\widehat{P_{n}}|\widehat{\psi}\rangle+\left(\frac{1}{r}-r\right)\right]=0.
\]
Since $\langle\widehat{\psi}|\widehat{P_{n}}|\widehat{\psi}\rangle\geq0$, the bracketed expression is strictly positive when $r\in(0,1)$. Hence $\sin\phi=0$, and therefore $\Im(z)=r\sin\phi=0$.

For part (b), if $|z|=1$ but $\Im(z)\neq0$ then the above equation (which holds for any $|\widehat{\psi}\rangle$ in the null space of $\gamma(z)$) says that $\langle\widehat{\psi}|\widehat{P_{n}}|\widehat{\psi}\rangle=0$. 
\end{proof}

We are now ready to give the proof of the main technical lemma. 

\begin{proof}[Proof of \lem{1}]
We can use \lem{diagonalize_gamma} to write 
\begin{equation}
W(x)=\left(\prod_{i=1}^{m+n-n_{c}}e_{i}(x)\right)\prod_{c=1}^{n_c}\left(x\lambda_{c}-x^{2}-1\right)\quad\text{for }x\in\R\label{eq:W_for_real_x}
\end{equation}
where $e_{i}(x)=\langle\widehat{v_{i}}(x)|\gamma(x)|\widehat{v_{i}}(x)\rangle$ are holomorphic functions of $x$ for $x\in\R$. The above expression explicitly separates out the contribution of the confined bound states to the determinant.

We now show that $\alpha_{1}=n_{b}+\dim\Conf_{>}$. Part (a) of \lem{wzero} shows that all of the roots of $W$ with magnitude less than $1$ are real. Consider one such root $x_{0}\in(-1,1)$. If $x_{0}$ is a root of the polynomial 
\[
x\lambda_{c}-x^{2}-1
\]
then the other root is $\frac{1}{x_{0}}\notin(-1,1)$. So each confined bound state with energy $\lambda_{c}=x_{0}+\frac{1}{x_{0}}$ is responsible for a zero of multiplicity one in the polynomial $W(x)$ at $x=x_{0}$, and no other zeros in the interval $(-1,1)$. Now turn to the unconfined bound states. For each $i=1,\ldots,m+n-n_{c}$, using \lem{deriv} we have 
\begin{align*}
\frac{de_{i}}{dx}\bigg|_{x_{0}} & =\left(\frac{1}{x_{0}}-x_{0}\right)+x_{0}\langle\widehat{v_{i}}(x_{0})|\widehat{P_{n}}|\widehat{v_{i}}(x_{0})\rangle\\
 & \neq0
\end{align*}
for $x_{0}\in(-1,1)$ since the right-hand side has the same sign as $x_{0}$ (and $x_{0}\neq0$ because $\gamma(0)=-1$). So each state $|\widehat{v_{i}}(x_{0})\rangle$ such that $e_{i}(x_{0})=0$ contributes a zero of multiplicity one at $x_{0}$ to the polynomial $W(z)$. Therefore, for each $x_{0}\in(-1,1)$ such that $W(x_{0})=0$, 
\begin{align*}
\text{multiplicity of the zero at $x_{0}$} & = \dim\left(\spn\left\{ |\psi\rangle \colon |\psi\rangle\in \Conf_{>}^{\perp} \text{ and } \gamma(x_{0})|\widehat{\psi}\rangle=0\right\} \right)\\
 &\quad +\dim\left(\spn\left\{ |\psi\rangle \colon |\psi\rangle\in \Conf_{>} \text{ and }\gamma(x_{0})|\widehat{\psi}\rangle=0\right\} \right).
\end{align*}
Hence 
\begin{align*}
\alpha_{1} & =\sum_{\substack{x_{0}\in(-1,1) \\ W(x_{0})=0}} \bigg[\dim\left(\spn\left\{ |\psi\rangle \colon |\psi\rangle\in \Conf_{>}^{\perp} \text{ and } \gamma(x_{0})|\widehat{\psi}\rangle=0\right\} \right)\\
 & \qquad+\dim\left(\spn\left\{ |\psi\rangle \colon |\psi\rangle\in \Conf_{>} \text{ and }\gamma(x_{0})|\widehat{\psi}\rangle=0\right\} \right)\bigg]\\
 & =n_{b}+\dim \Conf_{>},
\end{align*}
where in the last line we have used the fact that confined bound states corresponding to different energies are linearly independent (since $D$ is Hermitian).

We now show that $\alpha_{3}=n_{h}+2\dim\Conf_{=}$. To understand the zeros of $W(x)$ for $x\in\{-1,1\}$, we again use equation \eq{W_for_real_x}. Each bound state in $\Conf_=$ (i.e., each confined bound state with energy $\pm2$) contributes a zero of order two located at $x=\pm1$, since \[ \pm 2x-x^{2}-1=-\left(x\mp1\right)^{2}. \] On the other hand, there can also exist states $|\widehat{v_{i}}(\pm1)\rangle$ such that $e_{i}(\pm1)=0$. These are half-bound states and satisfy (by \lem{deriv})
\[ \frac{de_{i}}{dx}\bigg|_{\pm1}=\pm1\langle\widehat{v_{i}}(\pm1)|\widehat{P_{n}}|\widehat{v_{i}}(\pm1)\rangle\neq0,
\]
so each such half-bound state contributes a zero of order 1. Hence
\begin{align*}
  \alpha_{3}
  & =\left|\left\{ i\colon z_{i}\in\left\{ -1,1\right\} \right\} \right|\\ & =\sum_{x=\pm1}\bigg[\dim\left(\spn\left\{ |\psi\rangle \colon |\psi\rangle\in \Conf_{=}^{\perp} \text{ and }\gamma(x)|\widehat{\psi}\rangle=0\right\} \right)\\
  & \qquad +2\dim\left(\spn\left\{ |\psi\rangle \colon |\psi\rangle\in \Conf_{=} \text{ and }\gamma(x)|\widehat{\psi}\rangle=0\right\} \right)\bigg]\\
  & =n_{h}+2\dim \Conf_{=}.
\end{align*}

Finally, we show that $\alpha_{2}=\dim \Conf_{<}$. Now we are concerned with roots of $W(z)$ on the unit circle but not on the real axis. We can no longer use the expression \eq{W_for_real_x}, so we now derive an alternate formula that can be used in this case. Let $\{|\psi_{c}\rangle\}$ be a basis of confined bound states as before. As in the proof of \lem{diagonalize_gamma}, write
\[
\gamma(z)
=\sum_{c=1}^{n_{c}}|\widehat{\psi}_{c}\rangle\langle\widehat{\psi}_{c}|\left(z\lambda-z^{2}-1\right)
+\left(1-\sum_{c=1}^{n_{c}}|\widehat{\psi}_{c}\rangle\langle\widehat{\psi}_{c}|\right)\gamma(z)\left(1-\sum_{c=1}^{n_{c}}|\widehat{\psi}_{c}\rangle\langle\widehat{\psi}_{c}|\right).
\]
From this expression we obtain
\begin{equation}
W(z)=\det(M(z))\prod_{c=1}^{n_c}\left(z\lambda_{c}-z^{2}-1\right)\label{eq:Separate_confined-1}
\end{equation}
where $M(z)$ is some $\left(m+n-n_{c}\right)\times\left(m+n-n_{c}\right)$ matrix. Let $z_{0}$ satisfy $|z_{0}|=1$ and $z_{0}\notin\{-1,1\}.$ From part (b) of \lem{wzero} we see that all the corresponding bound states are confined bound states (i.e., satisfy $\langle\psi|P_{n}|\psi\rangle=0$). This means that $\det(M(z_{0}))\neq0$, so the number of zeros of $W(z)$ at $z_{0}$ is the same as the number of zeros of the polynomial $\prod_{c=1}^{n_c}\left(z\lambda_{c}-z^{2}-1\right)$ at $z_{0}$. Each confined bound state with energy $\lambda_{c}=z_0+\frac{1}{z_{0}}$ contributes a simple zero at $z_{0}$ and a simple zero at $\frac{1}{z_{0}} = z_0^* \ne z_0$, corresponding to the two roots of the polynomial
\[
z\lambda_{c}-z^{2}-1.
\]
So the number of zeros of $W(z)$ in \{$|z|=1$ and $z\notin\{-1,1\}\}$ is twice the number of confined bound states in the subspace $\Conf_{<}$. In other words, 
\begin{align*}
2\alpha_{2} & =\left|\left\{ i\colon\left|z_{i}\right|=1 \text{ and }  z_{i}\notin\{-1,1\}\right\} \right|\\
 & =2 \sum_{\substack{z \colon W(z)=0, \\ |z|=1,~ z \ne \pm 1}} \dim\left(\spn\left\{ |\psi\rangle\colon\text{$\langle\psi|P_{n}|\psi\rangle=0$ and  $\gamma(z)|\widehat{\psi}\rangle=0$}\right\} \right) \\
 & = 2 \dim \Conf_{<}. \qedhere
\end{align*}
\end{proof}

\section{Open questions}

Two directions for future work suggested previously \cite{2011JMP....52h2102C} remain open. The first is to find an algorithmic use for Levinson's theorem in quantum computation. Another direction is to consider the inverse scattering problem, where the goal is to reconstruct as much information as possible about the obstacle $\widehat{G}$ from scattered waves.

\section*{Acknowledgments}

We thank Jeffrey Goldstone for sharing his notes\cite{goldstone} on the completeness of bound and scattering states in the case $n=1$.
This work was supported in part by MITACS, NSERC, the Ontario Ministry of Research and Innovation, and the US ARO/DTO.


\bibliographystyle{hplain}
\bibliography{refs}

\appendix
\section*{Appendix: Completeness of scattering and bound states\label{app:completeness}}

\renewcommand{\theequation}{A.\arabic{equation}}

In this appendix we prove \thm{completeness}, establishing that the scattering states $|\sc_{j}(k)\rangle$ for $j=1,\ldots,n$ along with the unconfined and confined bound states form a complete basis for the Hilbert space. In other words, for any two vertices $v$ and $w$ in the graph,
\[
\langle v|\left(\int_{-\pi}^{0}\frac{dk}{2\pi}|\sc_{j}(k)\rangle\langle\sc_{j}(k)|+\sum_{b=1}^{n_{b}}|\phi_{b}\rangle\langle\phi_{b}|+\sum_{c=1}^{n_{c}}|\psi_{c}\rangle\langle\psi_{c}|\right)|w\rangle=\delta_{vw}.
\]
This is a generalization of the $n=1$ case,\cite{goldstone} and our proof follows the same steps. The proof has three parts:

\begin{enumerate}
\item Since vertices $(r,q)$ and $(s,w)$ with $r,s\geq2$ (i.e., vertices outside the graph on the semi-infinite paths) have no overlap with confined bound states, i.e., 
\[
\langle r,q|\left(\sum_{c=1}^{n_{c}}|\psi_{c}\rangle\langle\psi_{c}|\right)|s,w\rangle=0,
\]
we prove
\begin{equation}
\langle r,q|\left(\int_{-\pi}^{0}\frac{dk}{2\pi}\sum_{j=1}^{n}|\sc_{j}(k)\rangle\langle\sc_{j}(k)|\right)|s,w\rangle=\delta_{rs}\delta_{qw}-\langle r,q|\left(\sum_{b=1}^{n_{b}}|\phi_{b}\rangle\langle\phi_{b}|\right)|s,w\rangle.
\label{eq:part1}
\end{equation}
 
\item Similarly, for $|r,q\rangle$ with $r\geq2$ and $|v\rangle$ with $v$ a vertex in the $(m+n)$-vertex graph $\widehat{G}$, we show that 
\[
\langle r,q|\left(\int_{-\pi}^{0}\frac{dk}{2\pi}\sum_{j=1}^{n}|\sc_{j}(k)\rangle\langle\sc_{j}(k)|\right)|v\rangle=-\langle r,q|\left(\sum_{b=1}^{n_{b}}|\phi_{b}\rangle\langle\phi_{b}|\right)|v\rangle.
\]

\item When $|v\rangle$ and $|w\rangle$ are both basis states corresponding to vertices in $\widehat{G}$, the confined bound states play a role. We show that 
\begin{equation}
\langle w|\left(\int_{-\pi}^{0}\frac{dk}{2\pi}\sum_{j=1}^{n}|\sc_{j}(k)\rangle\langle\sc_{j}(k)|\right)|v\rangle=\delta_{wv}-\langle w|\left(\sum_{b=1}^{n_{b}}|\phi_{b}\rangle\langle\phi_{b}|\right)|v\rangle-\langle w|\left(\sum_{c=1}^{n_{c}}|\psi_{c}\rangle\langle\psi_{c}|\right)|v\rangle.\label{eq:part3}
\end{equation}
\end{enumerate}

\subsection*{Part 1}

Here we consider vertices $(r,q)$ and $(s,w)$ outside the graph, so $r,s\in\{2,3,4,\ldots\}$. For $q=w$, we have
\begin{align*}
&\langle r,q|\left(\int_{-\pi}^{0}\frac{dk}{2\pi}\sum_{j=1}^{n}|\sc_{j}(k)\rangle\langle\sc_{j}(k)|\right)|s,q\rangle \\ &\quad=\int_{-\pi}^{0}\frac{dk}{2\pi}\bigg(\left[e^{-ikr}+e^{ikr}S_{qq}(e^{ik})\right]\left[e^{iks}+e^{-iks}S_{qq}(e^{ik})^{\ast}\right]
+\sum_{j\neq q}e^{ik(r-s)}S_{qj}(e^{ik})S_{qj}(e^{ik})^{\ast}\bigg)\\
&\quad =\int_{-\pi}^{0}\frac{dk}{2\pi}\left(2\cos(k(r-s))+e^{ik(r+s)}S_{qq}(e^{ik})+e^{-ik(r+s)}S_{qq}(e^{ik})^{\ast}\right)
\end{align*}
where we have used the fact that $S(e^{ik})$ is unitary. Now, using the fact that $S(e^{ik})=S(e^{-ik})^{\dagger}$ (which can be seen from equation \eq{sdaggerinv}), we get 
\begin{equation}
\langle r,q|\left(\int_{-\pi}^{0}\frac{dk}{2\pi}\sum_{j=1}^{n}|\sc_{j}(k)\rangle\langle\sc_{j}(k)|\right)|s,q\rangle=\delta_{rs}+\int_{-\pi}^{\pi}\frac{dk}{2\pi}e^{ik(r+s)}S_{qq}(e^{ik}).\label{eq:first_mat_element}
\end{equation}
Now consider the matrix element (for $q\neq w$)
\begin{align*}
&\langle r,q|\left(\int_{-\pi}^{0}\frac{dk}{2\pi}\sum_{j=1}^{n}|\sc_{j}(k)\rangle\langle\sc_{j}(k)|\right)|s,w\rangle \\ &\quad=\int_{-\pi}^{0}\frac{dk}{2\pi}\bigg(\left[e^{-ikr}+e^{ikr}S_{qq}(e^{ik})\right] e^{-iks}S_{wq}(e^{ik})^{\ast}
+ e^{ikr}S_{qw}(e^{ik}) \left[e^{iks}+e^{-iks}S_{ww}(e^{ik})^{\ast}\right] \\
&\qquad
+\sum_{j\notin\{q,w\}}e^{ik(r-s)}S_{qj}(e^{ik})S_{wj}(e^{ik})^{\ast}\bigg)\\
&\quad=\int_{-\pi}^{0}\frac{dk}{2\pi}\left[e^{-ik(r+s)}S_{wq}(e^{ik})^{\ast}+e^{ik(r+s)}S_{qw}(e^{ik})\right]\\
&\quad=\int_{-\pi}^{\pi}\frac{dk}{2\pi}e^{ik(r+s)}S_{qw}(e^{ik}).
\end{align*}
where again we have used the fact that $S(e^{ik})^{\dagger}=S(e^{ik})^{-1}=S(e^{-ik})$. Putting this together with equation \eq{first_mat_element}, we have 
\begin{equation}
\langle r,q|\left(\int_{-\pi}^{0}\frac{dk}{2\pi}\sum_{j=1}^{n}|\sc_{j}(k)\rangle\langle\sc_{j}(k)|\right)|s,w\rangle=\delta_{rs}\delta_{qw}+\int_{-\pi}^{\pi}\frac{dk}{2\pi}e^{ik(r+s)}S_{qw}(e^{ik}).\label{eq:part_1_mat_el}
\end{equation}
We now evaluate the second term on the right-hand side. Letting $z=e^{ik}$, we write this as a contour integral over the unit circle $\Gamma$:
\[
\int_{-\pi}^{\pi}\frac{dk}{2\pi}e^{ik(r+s)}S_{qw}(e^{ik})=\oint_{\Gamma}\frac{dz}{2\pi i}z^{r+s-1}S_{qw}(z).
\]
Since $S(z)$ is unitary on the unit circle, each of its columns and each of its rows is normalized to $1$. So $S_{qw}(z)$ does not have any poles when $z$ is on the unit circle. Furthermore, by \lem{S_is_meromorphic}, $S_{qw}(z)$ is a meromorphic function of $z$ and all of its poles inside the unit circle occur at values $z_{0}$ satisfying either $z_{0}=0$ or $\det(\gamma(z_{0}))=0$. By part (a) of \lem{wzero} this means that all of its poles inside the unit circle lie on the real axis. Using the residue theorem,
\[
\oint_{\Gamma}\frac{dz}{2\pi i}z^{r+s-1}S_{qw}(z)=\sum_{\text{residues }x_{0}\in(-1,1)}\left(\Res_{x_{0}}[x^{r+s-1}S_{qw}(x)]\right)
\]
where $\Res_{x_0}[f(x)]$ denotes the residue of $f(x)$ at $x=x_0$. From \eq{gammahp}, we have 
\[
\gamma(\tfrac{1}{z})=\frac{1}{z^{2}}\gamma(z)+\left(\frac{1}{z^{2}}-1\right)\widehat{P_{n}};
\]
using \eq{S_analytic_continuation}, this implies
\begin{align}
S_{qw}(z) & =\langle\widehat{1,q}|-\gamma(z)^{-1}\gamma(\tfrac{1}{z})|\widehat{1,w}\rangle\nonumber \\
 & =\langle\widehat{1,q}|\left[-\frac{1}{z^{2}}+\left(1-\frac{1}{z^{2}}\right)\gamma(z)^{-1}\widehat{P_{n}}\right]|\widehat{1,w}\rangle\nonumber \\
 & =-\frac{1}{z^{2}}\delta_{qw}+\left(1-\frac{1}{z^{2}}\right)\langle\widehat{1,q}|\gamma(z)^{-1}|\widehat{1,w}\rangle.\label{eq:S_mat_elements}
\end{align}
Thus
\[
\oint_{\Gamma}\frac{dz}{2\pi i}z^{r+s-1}S_{qw}(z)=\sum_{\text{residues }x_{0}\in(-1,1)} \Res_{x_{0}}\left[-x^{r+s-3}\delta_{qw}+x^{r+s-1}\left(1-\frac{1}{x^{2}}\right)\langle\widehat{1,q}|\gamma(x)^{-1}|\widehat{1,w}\rangle\right].
\]
Since we are considering vertices outside the graph, $r\geq2$ and $s\geq2$. Thus the first term has no residues, so 
\[
\oint_{\Gamma}\frac{dz}{2\pi i}z^{r+s-1}S_{qw}(z)=\sum_{\text{residues }x_{0}\in(-1,1)} \Res_{x_{0}}\left[x^{r+s-1}\left(1-\frac{1}{x^{2}}\right)\langle\widehat{1,q}|\gamma(x)^{-1}|\widehat{1,w}\rangle\right].
\]

We can use \lem{diagonalize_gamma} to write 
\begin{equation}
\gamma(x)^{-1}=\sum_{c=1}^{n_c}|\widehat{\psi}_{c}\rangle\langle\widehat{\psi}_{c}|\left(x\lambda_{c}-x^{2}-1\right)^{-1}+\sum_{i=1}^{m+n-n_{c}}\frac{1}{e_{i}(x)}|\widehat{v_{i}}(x)\rangle\langle\widehat{v_{i}}(x)|\text{ for }x\in(-1,1).\label{eq:gamma_inv_on_real_axis}
\end{equation}
The confined bound states $|\psi_{c}\rangle$ satisfy $\langle\psi_{c}|1,j\rangle=0$
for all $j\in\{1,\ldots,n\}$. There is a residue for each state $|\widehat{v_{i}}(x_{0})\rangle$
that satisfies $\gamma(x_{0})|\widehat{v_{i}}(x_{0})\rangle=0$. We use \lem{deriv} to evaluate these residues, giving
\begin{align}
\oint_{\Gamma}\frac{dz}{2\pi i}z^{r+s-1}S_{qw}(z)
& =\sum_{\text{residues }x_{0}\in(-1,1)} \Res_{x_{0}}\!\left[x^{r+s-1}\left(1-\frac{1}{x^{2}}\right)\sum_{i=1}^{m+n-n_{c}}\!\frac{1}{e_{i}(x)}\langle\widehat{1,q}|\widehat{v_{i}}(x)\rangle\langle\widehat{v_{i}}(x)|\widehat{1,w}\rangle\right] \nonumber \\
 & =\sum_{\substack{x_{0}\in(-1,1)\\ W(x_{0})=0}} \, \sum_{i\colon e_{i}(x_{0})=0} x_{0}^{r+s-1}\left(1-\frac{1}{x_{0}^{2}}\right) \frac{\langle\widehat{1,q}|\widehat{v_{i}}(x_{0})\rangle\langle\widehat{v_{i}}(x_{0})|\widehat{1,w}\rangle}{\left.\frac{de_{i}}{dx}\right|_{x_{0}}}\nonumber \\
 & =\sum_{\substack{x_{0}\in(-1,1)\\ W(x_{0})=0}} \, \sum_{i\colon e_{i}(x_{0})=0} x_{0}^{r+s-1}\left(1-\frac{1}{x_{0}^{2}}\right)\frac{\langle\widehat{1,q}|\widehat{v_{i}}(x_{0})\rangle\langle\widehat{v_{i}}(x_{0})|\widehat{1,w}\rangle}{\frac{1}{x_{0}}-x_{0}+x_{0}\langle\widehat{v_{i}}(x_{0})|\widehat{P_{n}}|\widehat{v_{i}}(x_{0})\rangle}.\label{eq:contour_part1}
\end{align}

As described in \sec{eigenstates}, each vector $|\widehat{v_{i}}(x_{0})\rangle$ such that $e_{i}(x_{0})=0$ for some $x_{0}\in(-1,1)$ corresponds to an unconfined bound state $|\phi_{b}\rangle=|\phi_{b(i,x_{0})}\rangle$ and vice versa. This related unconfined bound state agrees up to normalization with $|\widehat{v_{i}}(x_{0})\rangle$ within the graph, so
\[
|\widehat{\phi}_{b(i,x_{0})}\rangle=N_{i,x_{0}}|\widehat{v_{i}}(x_{0})\rangle;
\]
the amplitudes outside the graph are 
\begin{align*}
\langle y,j|\phi_{b(i,x_{0})}\rangle & = N_{i,x_{0}} \langle\widehat{1,j}|\widehat{v_{i}}(x_{0})\rangle x_{0}^{y-1}\text{ for }y\in\{2,3,4,\ldots\}\text{ and }j\in\{1,2,\ldots,n\}.
\end{align*}
From equation \eq{normalizing}, the normalizing constant is 
\begin{align*}
N_{i,x_{0}}
&= \left(\frac{\|\widehat{P_n}|\widehat{v_i}(x_0)\rangle\|^2}{1-x_0^2}
+ \|(1-\widehat{P_n})|\widehat{v_i}(x_0)\rangle\|^2 \right)^{-\frac{1}{2}}\\
&=\left(\frac{1-x_{0}^{2}}{1-x_{0}^{2}+x_{0}^{2}\langle\widehat{v_{i}}(x_{0})|\widehat{P_{n}}|\widehat{v_{i}}(x_{0})\rangle}\right)^{\frac{1}{2}}
\end{align*}
where we have used the fact that $\langle\widehat{v_{i}}(x_{0})|\widehat{v_{i}}(x_{0})\rangle=1$. Substituting in equation \eq{contour_part1}, we obtain
\begin{align*}
\oint_{\Gamma}\frac{dz}{2\pi i}z^{r+s-1}S_{qw}(z) 
& = -\sum_{\substack{x_{0}\in(-1,1)\\ W(x_{0})=0}} \, \sum_{i\colon e_{i}(x_{0})=0}
x_{0}^{r+s-2} N_{i,x_{0}}^{2}\langle\widehat{1,q}|\widehat{v_{i}}(x_{0})\rangle\langle\widehat{v_{i}}(x_{0})|\widehat{1,w}\rangle \\
 & =-\sum_{\substack{x_{0}\in(-1,1)\\ W(x_{0})=0}} \, \sum_{i\colon e_{i}(x_{0})=0}
 x_{0}^{r+s-2} \langle1,q|\phi_{b(i,x_{0})}\rangle\langle\phi_{b(i,x_{0})}|1,w\rangle \\
 & = -\sum_{\substack{x_{0}\in(-1,1)\\ W(x_{0})=0}} \, \sum_{i\colon e_{i}(x_{0})=0} \langle r,q|\phi_{b(i,x_{0})}\rangle\langle\phi_{b(i,x_{0})}|s,w\rangle\\
 & =-\langle r,q|\left(\sum_{b=1}^{n_{b}}|\phi_{b}\rangle\langle\phi_{b}|\right)|s,w\rangle.
\end{align*}
Plugging this into equation \eq{part_1_mat_el} gives equation \eq{part1}, as claimed.

\subsection*{Part 2}

We wish to evaluate 
\[
\int_{-\pi}^{0}\frac{dk}{2\pi}\sum_{j=1}^{n}\langle r,q|\sc_{j}(k)\rangle\langle\widehat{\sc_{j}}(k)|\widehat{v}\rangle
\]
where $|\widehat{v}\rangle$ corresponds to one of the vertices in $\widehat{G}$, and where $r\geq2$.

We have
\begin{align*}
\left(\widehat{H}-2\cos(k)\right)|\widehat{\sc_{j}}(k)\rangle & = -\sum_{l=1}^{n}|\widehat{1,l}\rangle\langle2,l|\sc_{j}(k)\rangle,
\end{align*}
so 
\begin{equation}
\langle\widehat{\sc_{j}}(k)|\widehat{v}\rangle=\sum_{l=1}^{n}\langle\sc_{j}(k)|2,l\rangle\langle\widehat{1,l}|\frac{-1}{\widehat{H}-2\cos(k)}|\widehat{v}\rangle\label{eq:hat_mat_element}
\end{equation}
and 
\begin{align*}
&\int_{-\pi}^{0}\frac{dk}{2\pi}\sum_{j=1}^{n}\langle r,q|\sc_{j}(k)\rangle\langle\widehat{\sc_{j}}(k)|\widehat{v}\rangle \\
&\quad=\int_{-\pi}^{0} \frac{dk}{2\pi}\sum_{j=1}^{n}\sum_{l=1}^{n}\left(e^{-ikr}\delta_{jq}+e^{ikr}S_{qj}(e^{ik})\right)\left(e^{2ik}\delta_{jl}+e^{-2ik}S_{lj}(e^{ik})^{\ast}\right) \langle\widehat{1,l}|\frac{-1}{\widehat{H}-2\cos(k)}|\widehat{v}\rangle \\
&\quad=\int_{-\pi}^{\pi}\frac{dk}{2\pi}\left(e^{ik(r-2)}\langle\widehat{1,q}|\frac{-1}{\widehat{H}-2\cos(k)}|\widehat{v}\rangle+e^{ik(r+2)}\sum_{l=1}^{n}S_{ql}(e^{ik})\langle\widehat{1,l}|\frac{-1}{\widehat{H}-2\cos(k)}|\widehat{v}\rangle\right),
\end{align*}
where we have used unitarity of $S$ as well as the fact that $S(e^{ik})^{\dagger}=S(e^{-ik}).$
We can convert this expression to a contour integral over the unit circle, giving
\begin{align*}
&\int_{-\pi}^{0}\frac{dk}{2\pi}\sum_{j=1}^{n}\langle r,q|\sc_{j}(k)\rangle\langle\widehat{\sc_{j}}(k)|\widehat{v}\rangle \\
&\quad= \oint_{\Gamma}\frac{dz}{2\pi i}\left(z^{r-3}\langle\widehat{1,q}|\frac{-1}{\widehat{H}-z-\frac{1}{z}}|\widehat{v}\rangle+\sum_{l=1}^{n}z^{r+1}S_{ql}(z)\langle\widehat{1,l}|\frac{-1}{\widehat{H}-z-\frac{1}{z}}|\widehat{v}\rangle\right) \\
&\quad=\oint_{\Gamma}\frac{dz}{2\pi i}\bigg[\left(z^{r-3}-z^{r-1}\right)\langle\widehat{1,q}|\frac{-1}{\widehat{H}-z-\frac{1}{z}}|\widehat{v}\rangle +z^{r+1}\left(\frac{1}{z^{2}}-1\right)\langle\widehat{1,q}|\gamma(z)^{-1}\widehat{P_{n}}\frac{1}{\widehat{H}-z-\frac{1}{z}}|\widehat{v}\rangle\bigg]
\end{align*}
where in the second line we have used equation \eq{S_mat_elements}.  Now by \eq{gammahp},
\[
  \widehat{P_n} = \frac{1}{z^2} \gamma(z) - \frac{1}{z}\left(\widehat{H} - z - \frac{1}{z}\right),
\]
so 
\begin{equation}
\gamma(z)^{-1}\widehat{P_{n}}=\frac{1}{z^{2}}-\frac{1}{z}\gamma(z)^{-1}\left(\widehat{H}-z-\frac{1}{z}\right).\label{eq:gamma_inv_P}
\end{equation}
Thus we have
\begin{equation}
\int_{-\pi}^{0}\frac{dk}{2\pi}\sum_{j=1}^{n}\langle r,q|\sc_{j}(k)\rangle\langle\widehat{\sc_{j}}(k)|\widehat{v}\rangle=\oint_{\Gamma}\frac{dz}{2\pi i}  z^{r}\left(1-\frac{1}{z^{2}}\right)\langle\widehat{1,q}|\gamma(z)^{-1}|\widehat{v}\rangle.\label{cont_integral}
\end{equation}

We now show that the integrand on the right-hand side has no poles on the unit circle, so that we can use the residue theorem to evaluate the contour integral. To see this, recall that $\gamma(z)$ is block diagonal in an orthonormal basis that includes the confined bound states as basis vectors, where the two blocks correspond to the subspaces $\mathcal{C}$ (the subspace of all confined bound states) and $\mathcal{C}^{\perp}$. Note that because $\langle\psi_c|1,q\rangle=0$,
\begin{align*}
z^{r}\left(1-\frac{1}{z^{2}}\right)\langle\widehat{1,q}|\gamma(z)^{-1}|\widehat{v}\rangle
=z^{r}\left(1-\frac{1}{z^{2}}\right)\langle\widehat{1,q}|\left(1-\sum_{c=1}^{n_{c}}|\widehat{\psi}_{c}\rangle\langle\widehat{\psi}_{c}|\right)\gamma(z)^{-1}\left(1-\sum_{c=1}^{n_{c}}|\widehat{\psi}_{c}\rangle\langle\widehat{\psi}_{c}|\right)|\widehat{v}\rangle
\end{align*}
Part (b) of \lem{wzero} says that if $\gamma(z)$ has determinant $0$ at some $z\in\Gamma\setminus\{\pm1\}$ then its null space is a subspace of $\mathcal{C}$. Hence the $\left(m+n-n_{c}\right)\times\left(m+n-n_{c}\right)$ block of $\gamma(z)$ restricted to $\mathcal{C}^{\perp}$ has nonzero determinant and is therefore invertible on $\Gamma\setminus\{\pm1\}$. Thus the integrand above has no poles on $\Gamma\setminus\{\pm1\}$. At the points $z=\pm1$, it may be that $\det\gamma(\pm1)=0$, but nevertheless
\[
\det\left(\frac{1}{1-z^{2}}\gamma(z)\right)
\]
is nonzero for $z\in\{\pm1\}$. This follows from \lem{deriv}, which implies that an eigenvector in the null space of $\gamma(\pm1)$ has a simple zero in its eigenvalue at $z=\pm1$. So we have shown that the integrand in equation \eqref{cont_integral} has no poles on the unit circle.

The poles of the integrand inside the unit circle occur on the real axis, where we can use the expression \eq{gamma_inv_on_real_axis} and \lem{deriv} to evaluate their residues (as we did in part 1). This gives
\begin{align*}
\int_{-\pi}^{0}\frac{dk}{2\pi}\sum_{j=1}^{n}\langle r,q|\sc_{j}(k)\rangle\langle\widehat{\sc_{j}}(k)|\widehat{v}\rangle & =\sum_{\substack{x_{0}\in(-1,1)\\ W(x_{0})=0}} \, \sum_{i\colon e_{i}(x_{0})=0} x_{0}^{r}\left(1-\frac{1}{x_{0}^{2}}\right) \frac{\langle\widehat{1,q}|\widehat{v_{i}}(x_{0})\rangle\langle\widehat{v_{i}}(x_{0})|\widehat{v}\rangle}{\frac{1}{x_{0}}-x_{0}+x_{0}\langle\widehat{v_{i}}(x_{0})|\widehat{P_{n}}|\widehat{v_{i}}(x_{0})\rangle} \\
 & =-\sum_{\substack{x_{0}\in(-1,1)\\ W(x_{0})=0}} \, \sum_{i\colon e_{i}(x_{0})=0} x_{0}^{r-1} \langle1,q|\phi_{b(i,x_{0})}\rangle\langle\phi_{b(i,x_{0})}|v\rangle \\
 & =-\langle r,q|\left(\sum_{b=1}^{n_{b}}|\phi_{b}\rangle\langle\phi_{b}|\right)|v\rangle
\end{align*}
which completes the proof of part 2.

\subsection*{Part 3}

To prove equation \eq{part3} it is sufficient to show that 
\[
\langle\chi_{a}|\left(\int_{-\pi}^{0}\frac{dk}{2\pi}\sum_{j=1}^{n}|\sc_{j}(k)\rangle\langle\sc_{j}(k)|\right)|\chi_{b}\rangle=\delta_{ab}-\langle\chi_{a}|\left(\sum_{b=1}^{n_{b}}|\phi_{b}\rangle\langle\phi_{b}|\right)|\chi_{b}\rangle-\langle\chi_{a}|\left(\sum_{c=1}^{n_{c}}|\psi_{c}\rangle\langle\psi_{c}|\right)|\chi_{b}\rangle
\]
for any orthonormal basis of $m+n$ states $\{|\chi_{a}\rangle\}$ that are superpositions of the $m+n$ basis states corresponding to vertices in the graph. We choose to work in the orthonormal basis of eigenstates of the matrix $\widehat{H}$, so 
\[
\begin{pmatrix}
A & B^{\dagger}\\
B & D
\end{pmatrix}|\widehat{\chi}_{a}\rangle=E_{a}|\widehat{\chi}_{a}\rangle,
\]
and the state $|\chi_{a}\rangle$ in the extended Hilbert space is simply equal to $|\widehat{\chi}_{a}\rangle$ in the graph and has zero amplitude elsewhere. Therefore
\[
\langle\chi_{a}|\left(\int_{-\pi}^{0}\frac{dk}{2\pi}\sum_{j=1}^{n}|\sc_{j}(k)\rangle\langle\sc_{j}(k)|\right)|\chi_{b}\rangle=\int_{-\pi}^{0}\frac{dk}{2\pi}\sum_{j=1}^{n}\langle\widehat{\chi}_{a}|\widehat{\sc_{j}}(k)\rangle\langle\widehat{\sc_{j}}(k)|\widehat{\chi}_{b}\rangle,
\]
and using equation \eq{hat_mat_element} we get
\begin{align}
\int_{-\pi}^{0}\frac{dk}{2\pi}\sum_{j=1}^{n}\langle\widehat{\chi}_{a}|\widehat{\sc_{j}}(k)\rangle\langle\widehat{\sc_{j}}(k)|\widehat{\chi}_{b}\rangle &=\int_{-\pi}^{0}\frac{dk}{2\pi}\sum_{j,q,l=1}^{n}\bigg[\langle\widehat{\chi}_{a}|\frac{1}{\widehat{H}-2\cos(k)}|\widehat{1,q}\rangle\langle\widehat{1,l}|\frac{1}{\widehat{H}-2\cos(k)}|\widehat{\chi}_{b}\rangle\nonumber \\
&\quad\times\left(e^{-2ik}\delta_{jq}+e^{2ik}S_{qj}(e^{ik})\right)\left(e^{2ik}\delta_{jl}+e^{-2ik}S_{lj}(e^{ik})^{\ast}\right)\bigg]\nonumber \\
&= \int_{-\pi}^{\pi}\frac{dk}{2\pi}\bigg[\langle\widehat{\chi}_{a}|\frac{1}{\widehat{H}-2\cos(k)}\widehat{P_{n}}\frac{1}{\widehat{H}-2\cos(k)}|\widehat{\chi}_{b}\rangle \nonumber \\
&\quad+ e^{4ik} \sum_{q,l=1}^{n}\langle\widehat{\chi}_{a}|\frac{1}{\widehat{H}-2\cos(k)}|\widehat{1,q}\rangle S_{ql}(e^{ik}) \langle\widehat{1,l}|\frac{1}{\widehat{H}-2\cos(k)}|\widehat{\chi}_{b}\rangle \bigg]\nonumber \\
&=\oint_{\Gamma}\frac{dz}{2\pi iz}\bigg[\langle\widehat{\chi}_{a}|\frac{1}{\widehat{H}-z-\frac{1}{z}}\widehat{P_{n}}\frac{1}{\widehat{H}-z-\frac{1}{z}}|\widehat{\chi}_{b}\rangle \nonumber\\
&\quad + z^{4} \sum_{q,l=1}^{n}\langle\widehat{\chi}_{a}|\frac{1}{\widehat{H}-z-\frac{1}{z}}|\widehat{1,q}\rangle S_{ql}(z) \langle\widehat{1,l}|\frac{1}{\widehat{H}-z-\frac{1}{z}}|\widehat{\chi}_{b}\rangle \bigg].\label{eq:contour_part3}
\end{align}
By equation \eq{S_analytic_continuation},
\begin{align*}
S_{ql}(z) & = -\langle\widehat{1,q}|\widehat{P_{n}}\gamma(z)^{-1}\gamma(\tfrac{1}{z})\widehat{P_{n}}|\widehat{1,l}\rangle.
\end{align*}
Similarly to equation \eq{gamma_inv_P}, we have
\[
  \widehat{P_n} \gamma(z)^{-1} = \frac{1}{z^2} - \frac{1}{z}\left(H-z-\frac{1}{z}\right) \gamma(z)^{-1},
\]
and
\begin{align*}
  \gamma(\tfrac{1}{z})
  &= \frac{1}{z} \left(H-z-\frac{1}{z}\right) + \frac{1}{z^2} \widehat{P_n} \\
  &= \frac{1}{z^2} \gamma(z) + \left(\frac{1}{z^2}-1\right) \widehat{P_n},
\end{align*}
so
\begin{align*}
&\widehat{P_{n}}\gamma(z)^{-1}\gamma(\tfrac{1}{z})\widehat{P_{n}} \\ &\quad=\left[\frac{1}{z^{2}}-\frac{1}{z}\left(\widehat{H}-z-\frac{1}{z}\right)\gamma(z)^{-1}\right]\gamma(\tfrac{1}{z})\widehat{P_{n}} \\
 &\quad =\frac{1}{z^{2}}\left[\frac{1}{z}\left(\widehat{H}-z-\frac{1}{z}\right)+\frac{1}{z^{2}}\widehat{P_{n}}\right]\widehat{P_{n}}
 -\frac{1}{z}\left(\widehat{H}-z-\frac{1}{z}\right)\gamma(z)^{-1}\left[\frac{1}{z^{2}}\gamma(z)+\left(\frac{1}{z^{2}}-1\right)\widehat{P_{n}}\right]\widehat{P_{n}}\\
 &\quad =\frac{1}{z^{4}}\widehat{P_{n}}-\frac{1}{z}\left(\frac{1}{z^{2}}-1\right)\left(\widehat{H}-z-\frac{1}{z}\right)\gamma(z)^{-1}\widehat{P_{n}} \\
&\quad=\frac{1}{z^{4}}\widehat{P_{n}}-\frac{1}{z^{3}}\left(\frac{1}{z^{2}}-1\right)\left(\widehat{H}-z-\frac{1}{z}\right)+\frac{1}{z^{2}}\left(\frac{1}{z^{2}}-1\right)\left(\widehat{H}-z-\frac{1}{z}\right)\gamma(z)^{-1}\left(\widehat{H}-z-\frac{1}{z}\right)
\end{align*}
where in the last step we have used equation \eq{gamma_inv_P}. Inserting this into equation \eq{contour_part3} gives
\begin{align}
&\int_{-\pi}^{0}\frac{dk}{2\pi}\sum_{j=1}^{n}\langle\widehat{\chi}_{a}|\widehat{\sc_{j}}(k)\rangle\langle\widehat{\sc_{j}}(k)|\widehat{\chi}_{b}\rangle \nonumber\\
&\quad=\oint_{\Gamma}\frac{dz}{2\pi i}\left[\left(\frac{1}{z^{2}}-1\right)\langle\widehat{\chi}_{a}|\frac{1}{\widehat{H}-z-\frac{1}{z}}|\widehat{\chi}_{b}\rangle-z\left(\frac{1}{z^{2}}-1\right)\langle\widehat{\chi}_{a}|\gamma(z)^{-1}|\widehat{\chi}_{b}\rangle\right].\label{eq:contour_integral_part3}
\end{align}

Without loss of generality we assume the orthonormal basis $\{|\widehat{\chi}_{a}\rangle\}$ includes the confined bound states $\{|\widehat{\psi}_{c}\rangle\}$ since they are eigenstates of $\widehat{H}$. By equation \eq{gammacbs},
\[
\gamma(z)^{-1}|\widehat{\psi}_{c}\rangle=\frac{1}{z}\frac{1}{\lambda_{c}-z-\frac{1}{z}}|\widehat{\psi}_{c}\rangle.
\]
Plugging this into the above, we see that if either $|\widehat{\chi}_{a}\rangle$ or $|\widehat{\chi}_{b}\rangle$ corresponds to a confined bound state then the integrand is zero.

Now consider the case where both $|\widehat{\chi}_{a}\rangle$ and $|\widehat{\chi}_{b}\rangle$ are orthogonal to all confined bound states. Then the first term in equation \eq{contour_integral_part3} is 
\begin{align*}
\oint_{\Gamma}\frac{dz}{2\pi i}\left[\left(\frac{1}{z^{2}}-1\right)\frac{\delta_{ab}}{E_{a}-z-\frac{1}{z}}\right] & =\int_{-\pi}^{\pi}\frac{dk}{2\pi}\frac{-2\sin(k)\delta_{ab}}{E_{a}-2\cos(k)}
 =0
\end{align*}
since the integrand is an odd function of $k$.
For the second term in equation 2, since the states $|\widehat{\chi}_{a}\rangle$ and $|\widehat{\chi}_{b}\rangle$
are orthogonal to all confined bound states, the integrand
\begin{align*}
z\left(1-\frac{1}{z^{2}}\right)\langle\widehat{\chi}_{a}|\gamma(z)^{-1}|\widehat{\chi}_{b}\rangle & = z\left(1-\frac{1}{z^{2}}\right)\langle\widehat{\chi}_{a}|\left(1-\sum_{c=1}^{n_{c}}|\widehat{\psi}_{c}\rangle\langle\widehat{\psi}_{c}|\right)\gamma(z)^{-1}\left(1-\sum_{c=1}^{n_{c}}|\widehat{\psi}_{c}\rangle\langle\widehat{\psi}_{c}|\right)|\widehat{\chi}_{b}\rangle
\end{align*}
has no poles on the unit circle by the same argument given in part
2 for the integrand of equation \eqref{cont_integral}.  As in parts 1 and 2, the second term in equation \eq{contour_integral_part3} has poles inside the unit circle where $\det\gamma(z)=0$. However, unlike in parts 1 and 2, this term also has a pole at $z=0$ that must be taken into account. This pole has residue $-\langle\widehat{\chi}_{a}|\gamma(0)^{-1}|\widehat{\chi}_{b}\rangle=\delta_{ab}$, since $\gamma(0)=-1$. Computing the residues at all of the other poles proceeds as in parts 1 and 2, and we obtain 
\[
\oint_{\Gamma}\frac{dz}{2\pi i}\left[z\left(1-\frac{1}{z^{2}}\right)\langle\widehat{\chi}_{a}|\gamma(z)^{-1}|\widehat{\chi}_{b}\rangle\right]=\delta_{ab}-\langle\chi_{a}|\left(\sum_{b=1}^{n_{b}}|\phi_{b}\rangle\langle\phi_{b}|\right)|\chi_{b}\rangle.
\]
Hence 
\begin{align*}
&\int_{-\pi}^{0}\frac{dk}{2\pi}\sum_{j=1}^{n}\langle\widehat{\chi}_{a}|\widehat{\sc_{j}}(k)\rangle\langle\widehat{\sc_{j}}(k)|\widehat{\chi}_{b}\rangle \\
&\quad=\begin{cases}
0 & \text{if either $|\chi_{a}\rangle$ or $|\chi_{b}\rangle$ is a confined bound state}\\
\delta_{ab}-\langle\chi_{a}|\left(\sum_{b=1}^{n_{b}}|\phi_{b}\rangle\langle\phi_{b}|\right)|\chi_{b}\rangle & \text{if $\langle\widehat{\chi}_{a}|\widehat{\psi}_{c}\rangle=\langle\widehat{\chi}_{b}|\widehat{\psi}_{c}\rangle=0$ for all $c\in\{1,\ldots,n_{c}\}.$}
\end{cases}
\end{align*}
Equivalently,
\[
\langle\chi_{a}|\left[\int_{-\pi}^{0}\frac{dk}{2\pi}\sum_{j=1}^{n}|\sc_{j}(k)\rangle\langle\sc_{j}(k)|\right]|\chi_{b}\rangle=\delta_{ab}-\langle\chi_{a}|\left(\sum_{b=1}^{n_{b}}|\phi_{b}\rangle\langle\phi_{b}|\right)|\chi_{b}\rangle-\langle\chi_{a}|\left(\sum_{c=1}^{n_{c}}|\psi_{c}\rangle\langle\psi_{c}|\right)|\chi_{b}\rangle,
\]
which is what we wanted to prove.
\end{document}